\documentclass[letterpaper]{article}

\usepackage[a4paper,includefoot,margin=3cm]{geometry}
\usepackage[utf8]{inputenc}
\usepackage[T1]{fontenc}
\usepackage{microtype}
\usepackage{amsmath, amssymb, amsthm}
\usepackage{nicefrac}
\usepackage[mathic]{mathtools}
\usepackage[dvipsnames]{xcolor}
\usepackage{subcaption}
\usepackage{multirow}
\usepackage{authblk}
\usepackage{tabularx}
\usepackage{booktabs}
\usepackage{multirow}
\usepackage{dsfont}
\usepackage[]{todonotes}
\usepackage{multicol}
\usepackage{enumerate}

\usepackage{lineno}

\usepackage{tikz}
\usetikzlibrary{arrows, arrows.meta, calc, decorations.pathmorphing, decorations.pathreplacing, calligraphy, backgrounds, math, shapes, shapes.geometric, spy, positioning, quotes}
\usetikzlibrary{fit}
\usepackage{forest}
\useforestlibrary{linguistics} 

\tikzstyle{treeNode}=[circle,draw=black,inner sep=0pt,minimum size=12pt,font=\small]
\tikzstyle{steinerNode}=[circle,draw=black,fill=black,inner sep=0pt,minimum size=3pt]
\tikzstyle{triangleNode}=[isosceles triangle,rotate=270,draw=black,inner sep=0pt,minimum size=3pt]

\tikzstyle{treeA}=[thick,italyRed] 
\tikzstyle{treeB}=[thick,italyGreen] 
\tikzstyle{treeC}=[thick,blue] 
\tikzstyle{treeD}=[thick,orange] 
\tikzstyle{treeE}=[thick,black]

\usepackage{algorithm2e}

\usepackage{authblk}

\usepackage{physics}

\usepackage[pdfdisplaydoctitle,menucolor=orange!40!black,filecolor=magenta!40!black,urlcolor=blue!40!black,linkcolor=red!40!black,citecolor=green!40!black,ocgcolorlinks]{hyperref} 
\usepackage[capitalise]{cleveref}

\usepackage{thmtools}
\theoremstyle{plain}
\newtheorem{theorem}{Theorem}
\crefname{theorem}{Theorem}{Theorems}
\Crefname{theorem}{Thm{.}}{Thms{.}}
\newtheorem{lemma}[theorem]{Lemma}

\Crefname{observation}{Obs{.}}{Obs{.}}
\crefname{observation}{Observation}{Observations}

\theoremstyle{definition}

\theoremstyle{remark}

\definecolor{italyGreen}{RGB}{0, 146, 70}
\definecolor{italyRed}{RGB}{206, 43, 55}

\DeclarePairedDelimiter{\ceil}{\lceil}{\rceil}
\DeclarePairedDelimiter{\floor}{\lfloor}{\rfloor}

\DeclareMathOperator{\cost}{cost}
\DeclareMathOperator{\depth}{depth}
\DeclareMathOperator{\dist}{dist}

\newcommand{\opt}{\textnormal{\textrm{OPT}}}

\newcommand{\OptproblemDef}[3]{%
\begin{center}
	\setlength{\tabcolsep}{2pt}
	\begin{tabular}{@{}lp{12cm}@{}}
		\multicolumn{2}{@{}l}{\textsc{#1}} \\
		\textbf{Input:} & #2 \\
		\textbf{Task:} & #3 \\
	\end{tabular}
\end{center}%
}

\newcommand{\tibt}{\textnormal{\textsc{Tree in Binary Tree}}}
\newcommand{\tibtShort}{\textnormal{\textsc{TBT}}}

\newcommand{\tibtSteiner}{\textnormal{\textsc{Tree in Binary Tree With Steiner Nodes}}}
\newcommand{\tibtSteinerShort}{\textnormal{\textsc{TBT-SN}}}

\newcommand{\first}{\textnormal{\textsc{Bracket Builder}}}
\newcommand{\second}{\textnormal{\textsc{Tournament Runner}}}

\title{Designing Approximate Binary Trees for Trees}

\author{Leon Kellerhals}
\affil{Computational Intelligence Research Group, Technische Universität Clausthal}
\author{Mitja Krebs}
\affil{Intelligent Networks, Technische Universität Berlin}
\author{André Nichterlein}
\affil{Algorithmics and Computational Complexity, Technische Universität Berlin}
\author{Stefan Schmid}
\affil{Intelligent Networks, Technische Universität Berlin, Germany\\ Fraunhofer SIT}

\date{}

\begin{document}
\maketitle


\begin{abstract}
	We study the following problem that is motivated by demand-aware network design:
	Given a tree~$G$, the task is to find a binary tree~$H$ on the same vertex set.
	The objective is to minimize the sum of distances in~$H$ between vertex pairs that are adjacent in~$G$.
	We present a linear-time factor-4 approximation for this problem.
%
\end{abstract}

\section{Introduction}
\label{sec:intro}

Line embedding problems\footnote{Also known as linear arrangement or linear layout problems.} have been studied for over half a century due to their relevance in applications such as VLSI circuit design~\cite{diaz2002survey}.
Given a graph $G$, the task is to embed its vertices in a path $H$ so as to minimize a specific objective.
Studied objectives include the \emph{cutwidth},\footnote{The cutwidth is the smallest number~$k$ such that each edge in $H$ is on at most $k$ paths between vertex pairs that are adjacent in~$G$.} the \emph{bandwidth},\footnote{The bandwidth is the maximum distance in~$H$ between any two vertices adjacent in~$G$.} and the \emph{linear arrangement cost}.
The latter is the objective of interest in this work: The goal is to minimize the sum of edge lengths, i.e., the distances in $H$ between endpoints of edges in $G$.
All three objectives are NP-hard to minimize in general.
When~$G$ is a tree, minimizing bandwidth remains hard, whereas both cutwidth and linear arrangement cost can be minimized in polynomial time~\cite{chung1984optimal,CHUNG94}.

In the 1980s, variants of these problems---including embedding~$G$ into a \emph{binary tree}~$H$ instead of a path---were explored, especially for the cutwidth and bandwidth objectives~\cite{monien1985complexity,simonson1987variation}.
%
%
%
%
The corresponding problem for the linear arrangement cost objective has remained unexplored and is the focus of this work.
Specifically, we consider the following problem:
%
%

\OptproblemDef{\tibt{} (\tibtShort)}
{A tree~$G$.}
{Find a binary tree~$H$ with \(V(H)=V(G)\) that minimizes $\sum_{uv \in E(G)} \dist_H(u,v)$.}
For network embedding problems, the host graph~$H$ is usually a fixed part of the input~\cite{fischer2013vnep}.
Our setting fits better in the category of network \emph{design} problems---although earlier literature often grouped both under embedding~\cite{monien1985complexity,simonson1987variation}.

Recently, the network design view gained a new motivation by advances in optical switching technologies, which enable dynamically reconfigurable datacenter networks~\cite{cacm25,hall2021survey}.
These technologies allow the physical topology of a network to be adjusted to better align with communication demands, thereby reducing communication costs.
In this context, \(G\) represents the \emph{demand graph} (with unit-weight demands in our case), and the goal is to construct (design) a \emph{demand-aware} network \(H\) that minimizes the weighted communication distance---precisely the objective captured by \tibt{}.
While the general problem is of broader interest, our restriction to tree demands and topologies remains practically relevant:
tree-like communication patterns are common in parallel and distributed workloads, such as in machine learning~\cite{thakur2005optimization}, and tree topologies are widely used in real-world networks~\cite{al-fares2008scalable}.

\paragraph{Related work.}
%
As mentioned above, variants of \tibt{} have been studied already in the 1980s.
Monien~\cite{monien1985complexity} showed that, given a tree~$G$, finding a \emph{binary tree} that minimizes the bandwidth is NP-hard.
In contrast, Simonson~\cite{simonson1987variation} provided a linear-time algorithm that finds a binary tree minimizing the cutwidth.
However, the resulting binary tree may contain additional Steiner nodes, i.e., vertices not present in~$G$.
To the best of our knowledge, when restricted to binary trees without Steiner nodes, only a 3-approximation is known~\cite{simonson1987variation}.

More recently, demand-aware network design has attracted significant attention~\cite{cacm25,hall2021survey}.
We highlight here the most relevant contributions to our setting.

Avin et al.~\cite{avin2020demand} studied variants of \tibt{} where $G$ is sparse (i.e., with \(O(|V(G)|)\) edges) or~$H$ is a graph of bounded degree.
Their results include a constant-factor approximation algorithm for a variant of \tibt{} where \(H\) is a \emph{ternary} tree; we discuss their algorithm below in more detail.
%
More recently, Figiel et al.~\cite{opodis24dan} considered another problem variant that allows $H$ to contain Steiner nodes. 
When $G$ may be an arbitrary graph, they proved this variant to be NP-hard but admit a constant-factor approximation algorithm.
Notably, both works rely heavily on resource augmentation:
their algorithms are allowed more flexibility (additional nodes or higher degree in~$H$) than the optimum solution they are compared against.

The closest variant of \tibt{} known to be polynomial-time solvable was studied by Schmid et al.~\cite{splaynet}.
In their formulation, $G$ is an arbitrary graph with keys assigned to the vertices, and the task is to construct a binary \emph{search} tree (BST) that respects these keys.
This variant is solvable in polynomial time via dynamic programming.

Requiring~$H$ to be a BST, however, imposes structural constraints not present in \tibt.
As we show in \cref{ssec:bst}, even if~$G$ is a path on~$n$ vertices, the key assignment can force any BST solution to incur a cost that is~$\Omega(\log n)$ times higher than the optimal binary tree~$H=G$.

\paragraph{Our contribution.}
As opposed to minimizing bandwidth or cutwidth, the computational complexity of \tibt{} is unresolved.
While we are not able to determine its complexity, we provide a simple and elegant factor-4 approximation algorithm.
To the best of our knowledge, this is the first constant-factor approximation for the problem.

\begin{theorem}\label{thm:approximation}
	\tibt{} admits a linear-time \(4\)-approximation.
\end{theorem}

Our algorithm builds upon and extends existing approximation algorithms~\cite{avin2020demand,opodis24dan}.
Avin et al.~\cite{avin2020demand} construct for each vertex \(v\in V(G)\), a binary tree~\(T_{v}\) consisting of its children (assuming an arbitrary root if none is given) and connect \(v\) to the root of \(T_{v}\).
This results in a \emph{ternary} tree~\(H\):
a node~\(x\) may have two children within the binary subtree \(T_{v}\) corresponding to its parent \(v\) and the root of its own subtree \(T_{x}\) as a third child.

However, in our setting, where \(H\) must be a strictly \emph{binary} tree, it becomes impossible to satisfy both of the following desirable properties simultaneously:
\begin{enumerate}
	\item Each~$v \in V(G)$ induces a balanced subtree~$T_v$ consisting exclusively of the children of~$v$. \label{prop:only-children}
	\item Each~$v \in V(G)$ is either the root or directly adjacent to the root of its subtree~$T_v$.
\end{enumerate}
To address this conflict, Figiel et al.~\cite{opodis24dan} relax the first property by introducing Steiner nodes (as did Simonson~\cite{simonson1987variation} for minimizing the cutwidth): 
The subtree~$T_v$ introduced for vertex~$v$ is rooted at \(v\) and contains the children of \(v\) as leaves, while the inner nodes are Steiner nodes (see \Cref{fig:example-intro} (center) for an example).
\begin{figure}
	\begin{forest}
		[{},draw=white, s sep=12pt 
			[$v$,treeNode,for tree={s sep=8pt},edge={color=white},
				[$1$,treeNode,name=v1 [{},minimum size=0pt,inner sep=0pt,edge={color=white}]]
				[{$\ldots$},draw=white,edge={color=white}]
				[$7$,treeNode,name=v7 [{},minimum size=0pt,inner sep=0pt,edge={color=white}]]
			]
			[$v$,treeNode,for tree={s sep=8pt},edge={color=white},
				[{},steinerNode
				[{},steinerNode
					[{},steinerNode
						[$1$,treeNode [{},minimum size=0pt,inner sep=0pt [{$T_{1}$},roof]]]
						[$2$,treeNode [{},minimum size=0pt,inner sep=0pt [{$T_{2}$},roof]]]
					]
					[{},steinerNode
						[$3$,treeNode [{},minimum size=0pt,inner sep=0pt [{$T_{3}$},roof]]]
						[$4$,treeNode [{},minimum size=0pt,inner sep=0pt [{$T_{4}$},roof]]]
					]
				]
				[{},steinerNode
					[{},steinerNode
						[$5$,treeNode  [{},minimum size=0pt,inner sep=0pt [{$T_{5}$},roof]]]
						[$6$,treeNode [{},minimum size=0pt,inner sep=0pt [{$T_{6}$},roof]]]
					]
					[$7$,treeNode [{},minimum size=0pt,inner sep=0pt [{$T_{7}$},roof]]]
				]
				]
			]
			[$v$,treeNode,for tree={s sep=8pt},edge={color=white},
				[$1$,treeNode
				[$2$,treeNode
					[$3$,treeNode
						[{},minimum size=0pt,inner sep=0pt [{$T_{1}$},roof]]
						[{},minimum size=0pt,inner sep=0pt [{$T_{3}$},roof]]
					]
					[$4$,treeNode
						[{},minimum size=0pt,inner sep=0pt [{$T_{2}$},roof]]
						[{},minimum size=0pt,inner sep=0pt [{$T_{4}$},roof]]
					]
				]
				[$5$,treeNode
					[$6$,treeNode
						[{},minimum size=0pt,inner sep=0pt [{$T_{5}$},roof]]
						[{},minimum size=0pt,inner sep=0pt [{$T_{6}$},roof]]
					]
					[$7$,treeNode [{},minimum size=0pt,inner sep=0pt [{$T_{7}$},roof]]]
				]
				]
			]
		]
		\node [draw,below of=v1,yshift=1.25em,isosceles triangle,isosceles triangle apex angle=60,rotate=90] (t1) {};
		\node [draw,below of=v7,yshift=1.25em,isosceles triangle,isosceles triangle apex angle=60,rotate=90] (t2) {};
	\end{forest}
	\caption{
          Left:
          A tree \(G\) rooted at \(v\) with children \(1,\dots,7\), each with a subtree (depicted as triangles).
          Center:
          A binary tree \(H\) with Steiner nodes (small dots) violating the first property.
          Right:
          A binary tree \(H\) without Steiner nodes.
          Here, the second property is violated:
		for example, node $1$ is separated from its subtree \(T_{1}\) by two intermediate nodes, $2$ and $3$.
	}
	\label{fig:example-intro}
\end{figure}
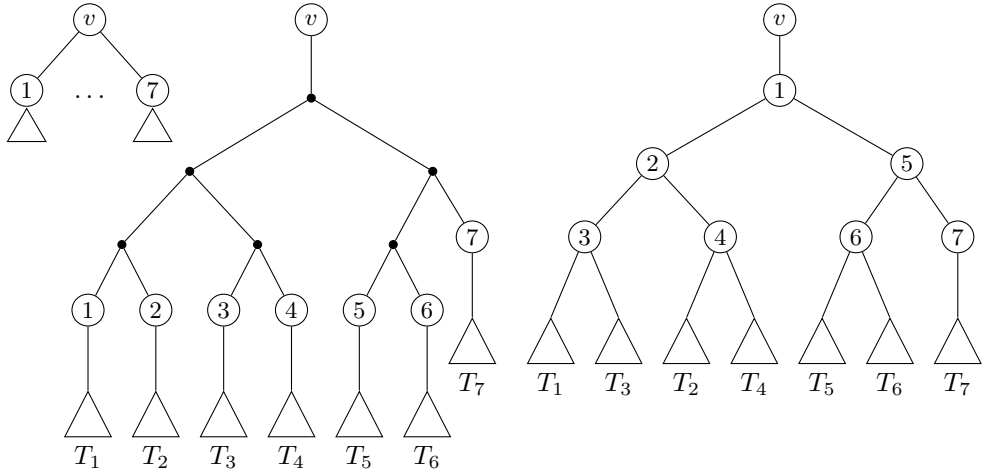
Constructing~$T_v$ as a balanced binary tree ensures that the average distance between \(v\) and its children increases only by a constant factor.
This approach---with minor modifications---forms the first phase of our algorithm, called \first{}.

Our main contribution lies in the second phase, \second{}, which eliminates all Steiner nodes, thereby re-establishing the first property at the cost of violating the second.
More precisely, each Steiner node in the subtree \(T_{v}\) is replaced by a child of \(v\), effectively detaching that child from its own subtree.
At a first glance, this
may appear to be a bad idea, as increasing the distance between a node \(x\) and its subtree \(T_{x}\) by even one incurs additional cost proportional to the number of children of~\(x\) in~\(G\) (see \Cref{fig:example-intro} (right) for an illustration).
However, we interpret each subtree \(T_{v}\) constructed by \first{} as a knockout tournament among the children of \(v\).
This novel perspective enables a nifty charging scheme that tightly bounds the additional cost caused by the elimination of Steiner nodes and yields the claimed approximation guarantee.

\paragraph{Notation.}
Let $G$ be a graph with $n$ vertices.
For a vertex $v \in V(G)$, we denote its degree by $\deg_G(v)$.
For a pair of vertices $u, v \in V(G)$, let $\dist_G(u,v)$ denote their shortest-path distance in $G$.
We omit the subscript when $G$ is clear from context.

Assume henceforth that $G$ is a tree. We fix an arbitrary root in $G$, thereby inducing a parent function $p \colon V(G) \to V(G)$.
For a vertex $v \in V(G)$, denote by $c_v$ the number of its children.

A \emph{complete} binary tree is defined as a binary tree in which each level, except possibly the last, is completely filled.
We define complete $k$-ary trees for $k> 2$ analogously.
We will use the fact that the height of a complete binary tree with \(n\) vertices is~\(\floor{\log n}\), and that of a complete binary tree with \(\ell\) leaves is \(\ceil{\log\ell}\).\footnote{Naturally, logarithms are base two unless stated otherwise.}

For a binary tree \(H\) with \(V(H)=V(G)\), we define the \emph{cost} of \(H\), i.e., the objective of \tibt{}, as~$\cost_{H}=\sum_{uv\in E(G)}\dist_{H}(u,v)$.
It is convenient to decompose this cost into contributions from each vertex \(u\in V(G)\) to its children in \(G\):
\[
  \cost_{H}(u)=\sum_{v:p(v)=u}\dist_{H}(u,v).
\]
Then the total cost can be expressed as
\(
  \cost_{H}=\sum_{u\in V(G)}\cost_{H}(u).
\)

Next to our main problem \tibt{}, we also consider a generalization
termed \tibtSteiner{} (\tibtSteinerShort).
In \tibtSteinerShort{}, the binary tree~$H$ may include additional \emph{Steiner nodes}, i.e., vertices not present in~$G$, so that $V(H) \supseteq V(G)$.

To avoid ambiguity, we refer to vertices in \(H\) as \emph{nodes} and to edges in \(H\) as \emph{links}.

\paragraph{Organization.}
First, in \cref{ssec:bst}, we discuss why binary search trees (BSTs) are not suitable for obtaining approximation algorithms.
In \cref{ssec:lower-bound}, we extend an argument from \cref{ssec:bst} to derive a lower bound for our approximation.
\cref{ssec:steiner-approx} is devoted to \first{}, our approximation algorithm for \tibtSteiner{}.
Finally, in \cref{ssec:tournament}, we discuss our main contribution: \second{}, a tournament-based approach for removing the Steiner nodes added by \first{}.

\section{How Not to Do It: The BST Approach}\label{ssec:bst}
Given that an optimal BST for a graph \(G\) with keys assigned to its vertices can be computed in polynomial time~\cite{splaynet}, one might naturally consider leveraging BSTs for approximation.
However, such an approach results in an approximation factor of~$\Omega(\log n)$.
To illustrate this, assume that \(G\) is a path with keys assigned in the following alternating order:
\[ 1,\;\, n/2+1,\;\, 2,\;\, n/2+2,\;\, \ldots,\;\, n-1,\;\, n/2,\;\, n. \]
In this construction, each vertex has a neighbor whose key differs by exactly $n/2$.
Notably, a binary tree structured identically to the path incurs a total cost of~$n-1$.
In contrast, we argue that any BST has cost \(\Omega(n \log n)\).

\begin{figure}[t!]
	\centering
	\begin{tikzpicture}
		\def\n{12} 
		\def\nHalfPO{7} 
		\def\nMO{11} 
		\foreach[count=\i] \num in {\nHalfPO,...,\n}
		{
			\node[treeNode] (v\i) at (2*\num,0) {$\i$};
			\node[treeNode] (v\num) at (2*\num + 1,0) {$\num$};
			\draw[-] (v\i) -- (v\num);
		}
		\foreach \num in {5,...,10}
		{
			\node[very thick, italyRed, draw,rounded corners, dashed,fit=(v\num)] {};
		}
		\foreach[count=\i from 2] \num in {\nHalfPO,...,\nMO}
		{
			\draw[-] (v\i) -- (v\num);
		}
	\end{tikzpicture}
	
	\bigskip
	
	\begin{forest}
		[$4$,for tree={treeNode,s sep=24pt}
			[2 [1][3]]
			[11 
				[8,for tree={s sep=12pt} [6,for tree={edge={treeA}} [5][7]] [9,for tree={edge={treeA}} [{},draw=white,edge={color=white}][10]] ]
				[12]
			]
		]
	\end{forest}
	\caption{
		An example for the construction in \cref{ssec:bst} with a path on~$n=12$ vertices (top) and a corresponding BST (bottom).
		Traversing from the root downward by always choosing the child with the larger subtree eventually leads to a subtree of size between~$n/4$ and~$n/2$.
		In this example, the subtree rooted at node~$8$ (highlighted in red) has size~$6=n/2$.
		Note that each node in this subtree has a neighbor on the path that lies outside the subtree.
	}
	\label{fig:bst-example}
\end{figure}
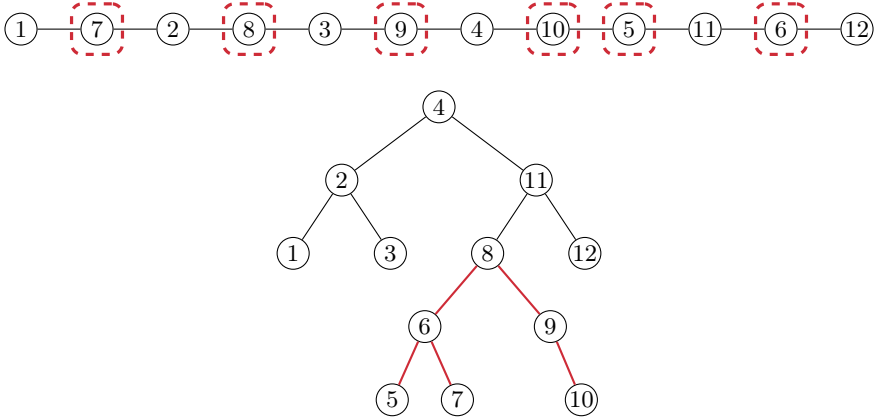

Observe that any BST with~$n$ nodes contains a subtree \(T\) of size between~$n/4$ and~$n/2$.
The keys in \(T\) necessarily form a contiguous subinterval of \(\{1,\dots,n\}\) with length at most~\(n/2\).
Due to the construction of the path, each node in \(T\) has a neighbor in \(G\) that lies outside \(T\).
Consequently, for each of the \(\Omega(n)\) such edges in \(G\), the corresponding path in the BST must traverse the root of \(T\).
Since the average distance from a node to the root in a binary tree is~$\Omega(\log n)$, the total cost is~$\Omega(n \log n)$, demonstrating that BSTs cannot achieve a constant-factor approximation in this setting.
See \cref{fig:bst-example} for an illustration.

\section{Lower Bound} \label{ssec:lower-bound}
In the above argument, we used the fact that in a binary tree with~$n$ nodes, the average distance from a node to the root is~$\Omega(\log n)$.
We now extend this observation to derive a simple lower bound on the total cost for general graphs \(H\) of bounded degree.
For any vertex \(v\in V(G)\) and any graph \(H\) with maximum degree \(\Delta\), the average distance in \(H\) between \(v\) and its neighbors in \(G\) is \(\Omega(\log_{\Delta-1}\deg_{G}(v))\).
This implies a lower bound of $\cost_{H}(v) \in \Omega(\deg_G(v) \log\deg_G(v))$ for any binary tree~$H$.
The overall lower bound is then obtained by summing over all vertices in~$G$.

To avoid double counting edges, we replace the degree~$\deg_G(v)$ with the number of children~$c_v$ of \(v\).
The following bounds suffice for our approximation:

\begin{lemma}[Lower Bound]\label{lem:lower-bound}
  Let \(H\) be a graph with \(V(H)\supseteq V(G)\) and maximum degree~\(\Delta\geq 3\), and let \(v\in V(G)\) be a vertex.
  Then, the cost of \(H\) satisfies the following bound:
  \begin{equation}
    \label{eq:lower-bound-1}
    \cost_{H}(v)\geq c_{v}h_{v}+\frac{\Delta h_{v}}{\Delta-2}-\frac{\Delta({(\Delta-1)}^{h_{v}}-1)}{{(\Delta-2)}^{2}},
  \end{equation}
  where \(h_{v}=1+\floor{\log_{\Delta-1}\ceil{c_{v}/\Delta}}\).

  Moreover, if \(c_{v}\geq \Delta(\Delta-1)\), this can be simplified as:
  \begin{equation}
    \label{eq:lower-bound-2}
    \cost_{H}(v)\geq c_{v}\left(\frac{\log c_{v}}{\log(\Delta-1)}-4\right).
  \end{equation}
\end{lemma}

\begin{proof}
  Consider the breadth-first search (BFS) tree \(T_{v}\) for node \(v\) in \(H\), constructed by halting BFS when all children of \(v\) in \(G\) are discovered.
  Then:
  \[
    \cost_{H}(v)=\sum_{w:p(w)=v}\depth_{T_{v}}(w).
  \]
  In the most favorable case, \(T_{v}\) comprises precisely \(v\) and its children and is a complete \((\Delta-1)\)-ary tree, where the root may have \(\Delta\) children, leading to:
  \[
    \cost_{H}(v)\geq \Delta\left(\sum_{\ell=1}^{h_{v}-1}{(\Delta-1)}^{\ell-1}\ell\right)+\left(c_{v}-\Delta\sum_{\ell=0}^{h_{v}-2}{(\Delta-1)}^{\ell}\right)h_{v},
  \]
  where \(h_{v}\) is the height of \(T_{v}\) in this best-case scenario, and the second addend accounts for all remaining nodes at the last level of \(T_{v}\).
  We recall the closed-form expressions for the geometric series and its first-order derivative:
  \[
    S_{n}=\sum_{k=0}^{n}r^{k}=\frac{r^{n+1}-1}{r-1}\quad\text{and}\quad\dv{r}S_{n}=\sum_{k=1}^{n}kr^{k-1}=\frac{(n+1)r^{n}(r-1)-(r^{n+1}-1)}{{(r-1)}^{2}}.
  \]
  Substituting these into the expression for \(\cost_{H}(v)\) gives:
  \[
    \cost_{H}(v)\geq \Delta\frac{h_{v}{(\Delta-1)}^{h_{v}-1}(\Delta-2)-{(\Delta-1)}^{h_{v}}+1}{{(\Delta-2)}^{2}}+\left(c_{v}-\Delta\frac{{(\Delta-1)}^{h_{v}-1}-1}{\Delta-2}\right)h_{v}.
  \]
  Simplifying and rearranging terms yields the bound in \cref{eq:lower-bound-1}.

  To determine the height \(h_{v}\) in the best-case scenario, notice that the root \(v\) connects to at most \(\Delta\) complete \((\Delta-1)\)-ary trees.
  Therefore, \(h_{v}=1+\floor{\log_{\Delta-1}\ceil{c_{v}/\Delta}}\).
  Substituting this and bounding the floor and ceiling terms, we get:
  \[
    \cost_{H}(v)\geq c_{v}\log_{\Delta-1}\left(\frac{c_{v}}{\Delta}\right)+\frac{\Delta\log_{\Delta-1}\left(\frac{c_{v}}{\Delta}\right)}{\Delta-2}-\frac{\Delta\left({\left(\Delta-1\right)}^{1+\log_{\Delta-1}\left(\frac{c_{v}}{\Delta}+1\right)}-1\right)}{{(\Delta-2)}^{2}}.
  \]
  This can be simplified to:
  \[
    \cost_{H}(v)\geq c_{v}\left(\log_{\Delta-1}\left(\frac{c_{v}}{\Delta}\right)-\frac{\Delta-1}{{(\Delta-2)}^{2}}\right)+\frac{\Delta}{\Delta-2}\left(\log_{\Delta-1}\left(\frac{c_{v}}{\Delta}\right)-1\right).
  \]
  If \(c_{v}\geq \Delta(\Delta-1)\), the second addend is non-negative and can be omitted.
  Changing the logarithm base yields:
  \[
    \cost_{H}(v)\geq c_{v}\left(\frac{\log c_{v}}{\log(\Delta-1)}-\frac{\Delta-1}{{(\Delta-2)}^{2}}-\frac{\log \Delta}{\log(\Delta-1)}\right).
  \]
  Finally, noting that for all \(\Delta\geq 3\), both \((\Delta-1)/{(\Delta-2)}^{2}\leq 2\) and \((\log \Delta)/(\log(\Delta-1))\leq 2\), we obtain \cref{eq:lower-bound-2}, completing the proof.
\end{proof}

\section{Approximation with Steiner Nodes}\label{ssec:steiner-approx}
%
%
%
The lower bound in \cref{lem:lower-bound} (for \(\Delta=3\)) naturally motivates a simple algorithm, referred to as \first{}, for constructing a binary tree \(H\) from a given tree~$G$.
We fix an arbitrary root in \(G\), thereby inducing a parent-child relation.
For each vertex \(v\in V(G)\), we construct a subtree \(T_{v}\) as follows (for clarity, we first assume that the number of children \(c_{v}\) of \(v\) is a power of two):
We build a complete binary tree \(T'_{v}\) whose leaves correspond to the children of \(v\) in \(G\), and whose inner nodes---including the root---are Steiner nodes.
We then connect~\(v\) to the root of \(T'_{v}\), forming \(T_{v}\).\footnote{We remark that we could also identify \(v\) with the root of \(T'_{v}\). The link between $v$ and the root, however, will be useful for the tournament phase.}
An illustration is provided in \cref{fig:steiner-approx}.

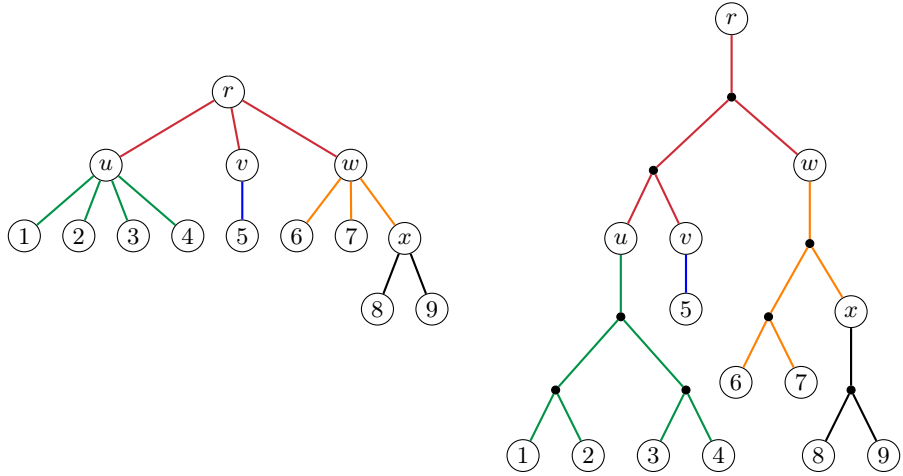
\begin{figure}[t]
	\centering
	\begin{forest}
		[{},draw=white, s sep=64pt 
		[{},draw=white, s sep=64pt,edge={color=white} 
		[$r$,for tree={treeNode,s sep=8pt},edge={color=white},name=origRoot 
			[$u$,for tree={edge={treeB}},edge={treeA},name=origU  [1][2][3][4]]
			[$v$,edge={treeA},name=origV [5,edge={treeC}]]
			[$w$,edge={treeA},name=origW [6,edge={treeD}][7,edge={treeD}][$x$,for tree={edge={treeE}},edge={treeD} [8][9]]]
		]
		]
		[$r$,for tree={treeNode,s sep=12pt},edge={color=white},name=newRoot 
			[{},steinerNode,edge={treeA}
				[{},steinerNode,edge={treeA}
					[$u$,for tree={edge={treeB}},edge={treeA},name=newU
						[{},steinerNode
							[{},steinerNode [1][2]]
							[{},steinerNode [3][4]]
						]
					]
					[$v$,edge={treeA},name=newV [5,edge={treeC}]]
				]
				[$w$,edge={treeA},name=newW 
					[{},steinerNode,edge={treeD}
						[{},steinerNode,edge={treeD} [6,edge={treeD}][7,edge={treeD}]]
						[$x$,for tree={edge={treeE}},edge={treeD} 
							[{},steinerNode,edge={treeE} [8][9]
								]
							]
					]
				]
			]
		]
		]
	\end{forest}
	\caption{
		Illustration of Algorithm \first{} for \tibtSteinerShort.
		The left shows the input tree~$G$; the right displays the binary tree~$H$ constructed by the algorithm.
                For each vertex~$v \in V(G)$, the corresponding subtree \(T_{v}\) in~$H$ is rooted at \(v\) and has the children of~$v$ in~$G$ as leaves.
		Subtrees are indicated by edge colors, and Steiner nodes are depicted as black dots.
	}
	\label{fig:steiner-approx}
\end{figure}

It is straightforward to see that this approach gives a constant-factor approximation:
By construction, each child~$w$ of a vertex~$v$ has distance~$\dist_H(v,w) = 1 + \log c_v$ from \(v\), as the complete binary tree \(T'_{v}\) has height~$\log c_v$ under the assumption that~$c_v$ is a power of two.
For arbitrary~$c_v$, we generalize the construction of \(T'_{v}\) by filling the last level so that all inner nodes have two children.
This modification yields the following bound:

\begin{lemma}[With Steiner Nodes]\label{lem:first-phase}
  Let \(H\) be the binary tree constructed by Algorithm \first{}, and let \(v\in V(G)\) be a vertex.
  Then:
  \[
    \cost_{H}(v)\leq c_{v}( \lceil \log c_{v} \rceil + 1).
  \]
\end{lemma}
%
Obtaining a constant-factor approximation from this upper bound is straightforward.
Determining the precise constant, however, requires particular attention for small values of~\(c_{v}\).

\begin{lemma}\label{lem:approximation-steiner}
	\tibtSteiner{} admits a linear-time \(3\)-approximation employing \(O(n)\) Steiner nodes.
\end{lemma}

\begin{proof}
  Let \opt{} be an optimum solution.
  We combine the upper bound from \cref{lem:first-phase} with the lower bounds from \cref{lem:lower-bound}.
  For any vertex \(v\) with \(c_{v}\geq 128\), we have:
  \[
    \cost_{H}(v)\leq c_{v}(\log c_{v}+2)\leq c_{v}(\log c_{v}+2)+2c_{v}\log c_{v}-14c_{v}=3c_{v}(\log c_{v}-4)\overset{\eqref{eq:lower-bound-2}}{\leq} 3\cost_{\opt}(v).
  \]
  For vertices with fewer children, i.e., \(c_{v}<128\), the following bound can be verified using \cref{tab:bound-ratios} in the appendix:
  \[
    \cost_{H}(v)\leq c_{v}(\ceil{\log c_{v}}+1)\leq 3\left(c_{v}h_{v}+3h_{v}-3(2^{h_{v}}-1)\right)\overset{\eqref{eq:lower-bound-1}}{\leq} 3\cost_{\opt}(v),
  \]
  where \(h_{v}=1+\floor{\log\ceil{c_{v}/3}}\).
  Summing over all vertices yields the approximation factor.

  As each introduced Steiner node has two children, the number of introduced Steiner nodes is bounded by the number of leaves in~$G$.
  The linear running time then follows trivially.
\end{proof}


\section{Running the Tournament: Approximation without Steiner Nodes}\label{ssec:tournament}
We now describe how to remove the Steiner nodes from the binary tree \(H\) constructed by Algorithm \first{}.
Let \(v\in V(G)\) be a vertex, and let \(T_{v}\) be the corresponding subtree in \(H\).
Recall that \(T_{v}\) is rooted at \(v\), its leaves correspond to the children of \(v\) in~\(G\), and its inner nodes are Steiner nodes.

Our algorithm, which we call \second{}, runs a knockout tournament for each subtree \(T_{v}\).
In this analogy, \(T_{v}\) serves as the tournament bracket:
the leaves represent the players, and each Steiner node corresponds to a \emph{match} between two players.

The \emph{winner} of a match between players \(x\) and \(y\) is the player with fewer children in \(G\), with ties broken arbitrarily, i.\,e., \(x\) wins only if \(c_{x}\le c_{y}\).
The winner replaces the Steiner node and advances to the next round, while the loser is eliminated from the tournament.
The loser inherits the subtree previously attached to the winner and is charged for the increase in cost caused by the winner's advancement.
Refer to \cref{fig:steiner-nodes-removal} for an illustration of a tournament and \cref{fig:example-tournament} for the result when applied to the binary tree of \cref{fig:steiner-approx}.

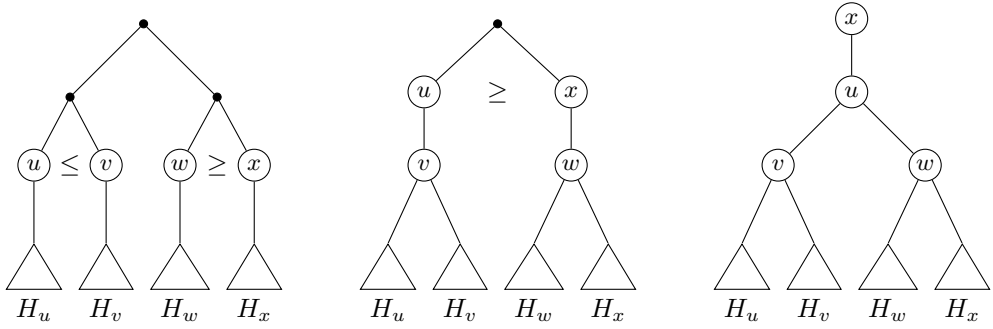
\begin{figure}[t!]
	\centering
	\begin{forest}
			[{},steinerNode
				[{},steinerNode
					[$u$,treeNode,name=u  [{},minimum size=0pt,inner sep=0pt [{$H_u$},roof]]]
					[$v$,treeNode,name=v [{},minimum size=0pt,inner sep=0pt [{$H_v$},roof]]]
				]
				[{},steinerNode
					[$w$,treeNode,name=w [{},minimum size=0pt,inner sep=0pt [{$H_w$},roof]]]
					[$x$,treeNode,name=x [{},minimum size=0pt,inner sep=0pt [{$H_x$},roof]]]
				]
			]
		{\node at ($(u)!0.5!(v)$) {$\le$};}
		{\node at ($(w)!0.5!(x)$) {$\ge$};}
	\end{forest}\hfil
	\begin{forest}
			[{},steinerNode
				[$u$,treeNode,name=u
					[$v$,treeNode 
						[{},minimum size=0pt,inner sep=0pt [{$H_u$},roof]]
						[{},minimum size=0pt,inner sep=0pt [{$H_v$},roof]] 
					]
				]
				[$x$,treeNode,name=x
					[$w$,treeNode 
						[{},minimum size=0pt,inner sep=0pt [{$H_w$},roof]]
						[{},minimum size=0pt,inner sep=0pt [{$H_x$},roof]]
					]
				]
			]
		{\node at ($(u)!0.5!(x)$) {$\ge$};}
	\end{forest}\hfil
	\begin{forest}
			[$x$,treeNode
				[$u$,treeNode
					[$v$,treeNode 
						[{},minimum size=0pt,inner sep=0pt [{$H_u$},roof]]
						[{},minimum size=0pt,inner sep=0pt [{$H_v$},roof]] 
					]
					[$w$,treeNode 
						[{},minimum size=0pt,inner sep=0pt [{$H_w$},roof]]
						[{},minimum size=0pt,inner sep=0pt [{$H_x$},roof]]
					]
				]
			]
	\end{forest}\hfil
	\caption{
		Illustration of Algorithm \second{}.
		Each Steiner node (small black dot) hosts a match between its two child nodes.
		The child with smaller degree in~$G$ wins and replaces the Steiner node, while the loser inherits the winner's subtree.
		In the leftmost tree, the siblings~$u,v,w,x$ satisfy~$\deg_G(u) \le \deg_G(v)$, $\deg_G(x) \le \deg_G(w)$, and~$\deg_G(x) \le \deg_G(u)$.
		The central tree shows the state after one round: $v$ and~$w$ are eliminated, while~$u$ and~$x$ proceed.
		The final tree shows~$x$ winning over~$u$.
	}
	\label{fig:steiner-nodes-removal}
\end{figure}

\begin{figure}[t!]
	\centering
	\begin{forest}
		[{},draw=white, s sep=16pt 
		[$r$,for tree={treeNode,s sep=8pt},edge={color=white},name=origRoot 
			[$u$,for tree={edge={treeB}},edge={treeA},name=origU  [1][2][3][4]]
			[$v$,edge={treeA},name=origV [5,edge={treeC}]]
			[$w$,edge={treeA},name=origW [6,edge={treeD}][7,edge={treeD}][$x$,for tree={edge={treeE}},edge={treeD} [8][9]]]
		]
		[$r$,for tree={treeNode,s sep=12pt},edge={color=white},name=steinerRoot 
			[{},steinerNode,edge={treeA}
				[{},steinerNode,edge={treeA}
					[$u$,for tree={edge={treeB}},edge={treeA},name=newU
						[{},steinerNode
							[{},steinerNode [1][2]]
							[{},steinerNode [3][4]]
						]
					]
					[$v$,edge={treeA},name=newV [5,edge={treeC}]]
				]
				[$w$,edge={treeA},name=newW 
					[{},steinerNode,edge={treeD}
						[{},steinerNode,edge={treeD} [6,edge={treeD}][7,edge={treeD}]]
						[$x$,for tree={edge={treeE}},edge={treeD} 
							[{},steinerNode,edge={treeE} [8][9]
								]
							]
					]
				]
			]
		]
		[$r$,for tree={treeNode,s sep=12pt},edge={color=white},name=finalRoot 
			[$v$,edge={treeA}
				[$w$,edge={treeA}
					[$u$,for tree={edge={treeB}},edge={treeA}
						[1,
							[3 [2][4]]
						]
						[5,edge={treeC}]
					]
					[6,edge={treeD}
						[$x$,for tree={edge={treeE}},edge={treeD} 
							[7,edge={treeD}]
							[8,[9]]
						]
					]
				]
			]
		]
		]
	\end{forest}
	\caption{
		Result of applying Algorithms \first{} and \second{} to the instance in \cref{fig:steiner-approx}.
		Ties are resolved lexicographically by the name (or number) of the node.
	}
	\label{fig:example-tournament}
\end{figure}
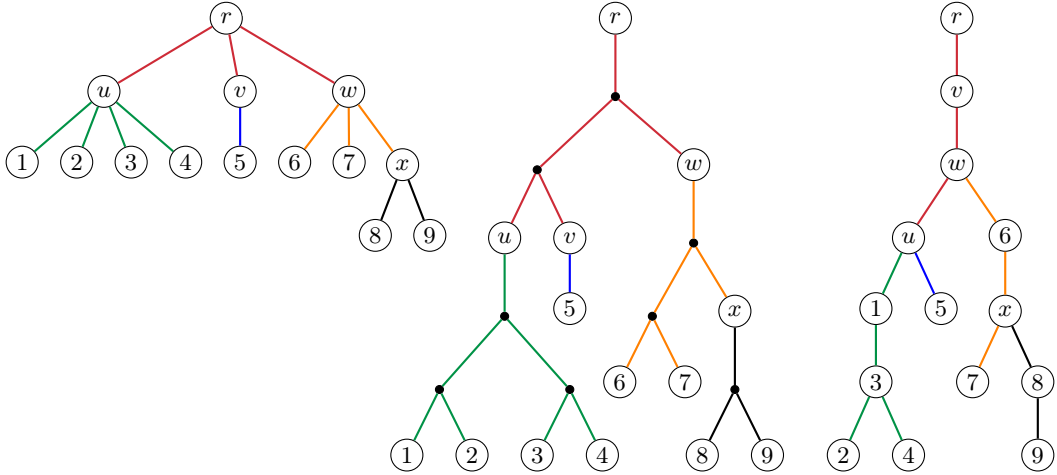

We now analyze the amount charged to a losing player.
As indicated above, whenever a winner \(x\) advances, the distance to each of its children increases by one, while all other distances do not increase.
The resulting cost increase is at most \(c_{x}\).
Since \(x\) wins only if~\(c_{x}\le c_{y}\), this increase can be fully charged to the losing node \(y\), which is charged \(c_{y}\).
Crucially, a player can lose at most one match and is thus charged at most once.

Notably, the subtree resulting from such a tournament can be constructed directly, without simulating the tournament explicitly:
Given a vertex, we construct a complete binary tree over its children in \(G\), placing those with smaller degree closer to the root.

Before establishing the approximation guarantee, we state three easy-to-see invariants of our algorithm.

\begin{lemma}[Invariants]\label{lem:invariants}
  The following invariants hold at the end of Algorithm \first{} and are maintained throughout Algorithm \second{}:
  \begin{enumerate}[(i)]
  \item\label{itm:invariant-1}
    Any Steiner node has two children.
  \item\label{itm:invariant-2}
    For any vertex \(v\in V(G)\), its parent in \(G\) is an ancestor of \(v\) in \(H\).
  \item\label{itm:invariant-3}
    For any vertex \(v\in V(G)\), if \(v\) is the child of a Steiner node in \(T_{u}\), then \(v\)'s parent in \(G\) is \(u\), and in \(H\), \(v\) either is a leaf or has a single child.
  \end{enumerate}
\end{lemma}

\begin{lemma}[Removing Steiner Nodes]\label{lem:second-phase}
  The total increase in cost incurred by Algorithm \second{} is at most \(n-1\).
\end{lemma}

\begin{proof}
  Consider a match between sibling nodes \(x\) and \(y\), with \(x\) winning.
  By \cref{lem:invariants}, advancing \(x\) increases its distance to each of its children by one, while all other distances do not increase.
  Hence, the cost increase due to the match is at most \(c_{x}\leq c_{y}\).

  Once a node loses a match, it becomes the child of a non-Steiner node and cannot participate in subsequent matches.
  As a result, each node \(v\) is charged \(c_{v}\) at most once.
  Consequently, the total cost increase for eliminating Steiner nodes in a subtree \(T_{u}\) is at most \(\sum_{v:p(v)=u}c_v\).
  Summing over all vertices completes the proof.
\end{proof}

\begin{proof}[Proof of \Cref{thm:approximation}]
  The approximation factor follows from \cref{lem:approximation-steiner}, \cref{lem:second-phase} and the fact that any graph \(H\) must incur a cost of at least \(|E(G)|=n-1\).
  The~$O(n)$ Steiner nodes result in~$O(n)$ matches.
  With~$O(1)$ time per match the linear running time follows.
\end{proof}

\section{Conclusion}
\label{sec:concl}
While our approximation for \tibt{} is both simple and concise, there remains room for improving the approximation factor.
We briefly outline three directions for refinement:
First, the upper bound on the cost incurred by Algorithm \first{} (\cref{lem:approximation-steiner}) assumes that all leaves are at the last level of the tree.
However, unless the number of children is a power of two, some leaves will in fact reside at the second-to-last level, leading to a tighter bound.
Second, in the current design of Algorithm \second{}, the overall tournament winner has only a single child.
Granting the winner a second child can reduce the cost of the solution.
Third, the current analysis of the second phase overlooks that the winner of a match moves closer to its parent in \(G\).
We conjecture that incorporating this detail in the analysis could reduce the additional cost from~$n-1$ to the number of leaves in~$G$.
Combining these observations with case distinctions for low-degree vertices in~$G$ (which are the current bottleneck in our analysis, see \cref{tab:bound-ratios} in the appendix) promises to reduce the approximation factor.
However, such refinements may come at the cost of the algorithm's elegance and substantially increase the complexity of the analysis.

Another open question is, whether (under standard complexity assumptions) one can improve the lower bound for the approximation factor; currently it is~$1$.

\bibliographystyle{plainurl}
\bibliography{References.bib}

\newpage

\appendix

\section{Appendix}

\begin{table}[h!]
	\caption{
		Bound ratios.
		For a number \(c_{u}\) of children, \(\textrm{LB}=c_{u}h_{u}+3h_{u}-3(2^{h_{u}}-1)\) is the lower bound on \(\cost_{H}\) proved in \cref{lem:lower-bound}, whereas \(\textrm{UB}=c_{u}(\ceil{\log c_{u}}+1)\) is the upper bound shown in \cref{lem:approximation-steiner}.
		\(\textrm{Ratio}=\textrm{UB}/\textrm{LB}\).
		Values are truncated after the second decimal digit.
	}\label{tab:bound-ratios}
	\small
	\begin{tabular}{rrrrp{2em}rrrrp{2em}rrrr}
		\toprule
		$c_{u}$ & LB & UB & Ratio&&$c_{u}$ & LB & UB & Ratio&&$c_{u}$ & LB & UB & Ratio\\
		\midrule
		 1	&1	&1	&1\phantom{.00}	&&44	&143	&308	&2.15		&&87	&357	&696	&1.94\\
		 2	&2	&4	&2\phantom{.00}	&&45	&147	&315	&2.14		&&88	&362	&704	&1.94\\
		 3	&3	&9	&3\phantom{.00}	&&46	&152	&322	&2.11		&&89	&367	&712	&1.94\\
		 4	&5	&12	&2.4\phantom{0}	&&47	&157	&329	&2.09		&&90	&372	&720	&1.93\\
		 5	&7	&20	&2.85		&&48	&162	&336	&2.07		&&91	&377	&728	&1.93\\
		 6	&9	&24	&2.66		&&49	&167	&343	&2.05		&&92	&382	&736	&1.92\\
		 7	&11	&28	&2.54		&&50	&172	&350	&2.03		&&93	&387	&744	&1.92\\
		 8	&13	&32	&2.46		&&51	&177	&357	&2.01		&&94	&393	&752	&1.91\\
		 9	&15	&45	&3\phantom{.00}	&&52	&182	&364	&2\phantom{.00}	&&95	&399	&760	&1.90\\
		10	&18	&50	&2.77		&&53	&187	&371	&1.98		&&96	&405	&768	&1.89\\
		11	&21	&55	&2.61		&&54	&192	&378	&1.96		&&97	&411	&776	&1.88\\
		12	&24	&60	&2.5\phantom{0}	&&55	&197	&385	&1.95		&&98	&417	&784	&1.88\\
		13	&27	&65	&2.40		&&56	&202	&392	&1.94		&&99	&423	&792	&1.87\\
		14	&30	&70	&2.33		&&57	&207	&399	&1.92		&&100	&429	&800	&1.86\\
		15	&33	&75	&2.27		&&58	&212	&406	&1.91		&&101	&435	&808	&1.85\\
		16	&36	&80	&2.22		&&59	&217	&413	&1.90		&&102	&441	&816	&1.85\\
		17	&39	&102	&2.61		&&60	&222	&420	&1.89		&&103	&447	&824	&1.84\\
		18	&42	&108	&2.57		&&61	&227	&427	&1.88		&&104	&453	&832	&1.83\\
		19	&45	&114	&2.53		&&62	&232	&434	&1.87		&&105	&459	&840	&1.83\\
		20	&48	&120	&2.5\phantom{0}	&&63	&237	&441	&1.86		&&106	&465	&848	&1.82\\
		21	&51	&126	&2.47		&&64	&242	&448	&1.85		&&107	&471	&856	&1.81\\
		22	&55	&132	&2.4\phantom{0}	&&65	&247	&520	&2.10		&&108	&477	&864	&1.81\\
		23	&59	&138	&2.33		&&66	&252	&528	&2.09		&&109	&483	&872	&1.80\\
		24	&63	&144	&2.28		&&67	&257	&536	&2.08		&&110	&489	&880	&1.79\\
		25	&67	&150	&2.23		&&68	&262	&544	&2.07		&&111	&495	&888	&1.79\\
		26	&71	&156	&2.19		&&69	&267	&552	&2.06		&&112	&501	&896	&1.78\\
		27	&75	&162	&2.16		&&70	&272	&560	&2.05		&&113	&507	&904	&1.78\\
		28	&79	&168	&2.12		&&71	&277	&568	&2.05		&&114	&513	&912	&1.77\\
		29	&83	&174	&2.09		&&72	&282	&576	&2.04		&&115	&519	&920	&1.77\\
		30	&87	&180	&2.06		&&73	&287	&584	&2.03		&&116	&525	&928	&1.76\\
		31	&91	&186	&2.04		&&74	&292	&592	&2.02		&&117	&531	&936	&1.76\\
		32	&95	&192	&2.02		&&75	&297	&600	&2.02		&&118	&537	&944	&1.75\\
		33	&99	&231	&2.33		&&76	&302	&608	&2.01		&&119	&543	&952	&1.75\\
		34	&103	&238	&2.31		&&77	&307	&616	&2.00		&&120	&549	&960	&1.74\\
		35	&107	&245	&2.28		&&78	&312	&624	&2\phantom{.00}	&&121	&555	&968	&1.74\\
		36	&111	&252	&2.27		&&79	&317	&632	&1.99		&&122	&561	&976	&1.73\\
		37	&115	&259	&2.25		&&80	&322	&640	&1.98		&&123	&567	&984	&1.73\\
		38	&119	&266	&2.23		&&81	&327	&648	&1.98		&&124	&573	&992	&1.73\\
		39	&123	&273	&2.21		&&82	&332	&656	&1.97		&&125	&579	&1000	&1.72\\
		40	&127	&280	&2.20		&&83	&337	&664	&1.97		&&126	&585	&1008	&1.72\\
		41	&131	&287	&2.19		&&84	&342	&672	&1.96		&&127	&591	&1016	&1.71\\
		42	&135	&294	&2.17		&&85	&347	&680	&1.95		&&	&	&	&\\
		43	&139	&301	&2.16		&&86	&352	&688	&1.95		&&	&	&	&\\
		\bottomrule
	\end{tabular}
\end{table}

%
%

\end{document}